\documentclass[sn-mathphys]{sn-jnl}

\jyear{2022}

\usepackage{lscape}
\usepackage{soul}
\usepackage{bussproofs}
\usepackage{amsthm}
\newtheorem{theorem}{Theorem}

\newcommand{\dimplus}{\ensuremath{\;\hat{+}\;}}
\newcommand{\tplus}{\ensuremath{\;\Diamond\;}}
\newcommand{\tplusNS}{\ensuremath{\Diamond}}
\newcommand{\dimtimes}{\ensuremath{\;\hat{\times}\;}}
\newcommand{\ttimes}{\ensuremath{\;\triangle\;}}
\newcommand{\tassign}[2]{\ensuremath{\; \lhd_{#2}^{#1}\;}}
\newcommand{\tassignNS}{\ensuremath{ \lhd}}
\newcommand{\strachey}[1]{\ensuremath{[\![ #1 ]\!]}}
\newcommand{\openS}{[\![ }
\newcommand{\closeS}{ ]\!]}

\begin{document}

\title{A Discipline of Programming with Quantities}

\author*{\fnm{Steve} \sur{McKeever}} \email{steve.mckeever@im.uu.se}
\affil{\orgdiv{Department of Informatics and Media}, \orgname{Uppsala University}, \orgaddress{\country{Sweden}}}


\abstract{
In scientific and engineering applications, physical quantities embodied as units of measurement (UoM) are frequently used.
The loss of the Mars climate orbiter, attributed to a confusion between the metric and imperial unit systems, 
popularised the disastrous consequences of incorrectly handling measurement values. 
Dimensional analysis can be used to ensure expressions containing annotated values are evaluated correctly.
This has led to the development of a large number of libraries, languages and validators to ensure developers can specify and verify
UoM information in their designs and codes. 
Many tools can also automatically convert values between commensurable UoM, such as yards and metres.
However these systems do not differentiate between quantities and dimensions. 
For instance torque and work, which share the same UoM, can not be interchanged because they do not represent the same entity.
We present a named quantity layer that complements dimensional analysis by ensuring that values of different quantities
are safely managed. 
Our technique is a mixture of analysis and discipline, where expressions involving multiplications are relegated to functions,
in order to ensure that named quantities are handled soundly.
}



\maketitle

\section{Introduction}
\label{intro}

\begin{table*}[t]
\begin{center}
\caption{Some SI standard base and derived units.}
\label{tab:derived}
{\small
\begin{tabular}{l|c|l|l}
Name & \verb!   !Symbol \verb!   !& Quantity \verb!   ! & \verb!   !Base Units in SI  \\
\hline 
metre & $l$ & \textsf{length} &\verb!   !$\mathsf{metre}$ \\
kilogram & $m$ & \textsf{mass} &\verb!   !$\mathsf{kg}$ \\
second & $t$ & \textsf{time} & \verb!   !$\mathsf{sec}$ \\
hertz	& $Hz$	& \textsf{frequency} &\verb!   !$\mathsf{sec}^{-1}$ \\
newton & $N$ & \textsf{force, weight} & \verb!   !$\mathsf{metre} \times \mathsf{kg} \times \mathsf{sec}^{-2}$ \\
pascal & $Pa$ & \textsf{pressure, stress} & \verb!   !$\mathsf{metre}^{-1} \times \mathsf{kg}\times \mathsf{sec}^{-2}$ \\
joule	& $J$	& \textsf{energy, work} & \verb!   !$\mathsf{metre}^{2} \times \mathsf{kg} \times \mathsf{sec}^{-2}$ \\
newton metre	& $N\;m$	& \textsf{torque} & \verb!   !$\mathsf{metre}^{2} \times \mathsf{kg} \times \mathsf{sec}^{-2}$ \\
watt	& $W$	& \textsf{power, flux} & \verb!   !$\mathsf{metre}^{2} \times \mathsf{kg} \times \mathsf{sec}^{-3}$ \\
square metre & $m^{2}$ & \textsf{area} & \verb!   !$\mathsf{metre}^{2}$ \\
cubic metre & $m^{3}$ & \textsf{volume} & \verb!   !$\mathsf{metre}^{3}$ \\
metre per second & $m/s$ & \textsf{speed, velocity} & \verb!   !$\mathsf{metre} \times \mathsf{sec}^{-1}$ \\
metre per sec squared	 & $m/s^{2}$ & \textsf{acceleration} & \verb!   !$\mathsf{metre} \times \mathsf{sec}^{-2}$ \\
\end{tabular}}
\end{center}
\end{table*}

Humans have used local units of measurement since the days of early trade, enhanced over time to fulfil the 
accuracy and interoperable needs of science and technology.
The technical definition of a physical quantity is a ``property of a phenomenon, body, or substance, where the 
property has a magnitude that can be expressed as a number and a reference''~\cite{vim2012}.
Ensuring numerical values that denote physical quantities are handled correctly is an essential
requirement for the design and development of any engineering application. 
Infamous examples such as the Mars Climate Orbiter~\cite{mars} or the Gimli Glider incident~\cite{flight143}
substantiate this. 
With ubiquitous digitalisation, and removal of humans in the loop, the need to faithfully represent
and manipulate quantities in physical systems is ever increasing. Programming languages allow developers to describe how 
to evaluate numeric expressions but not how to detect inappropriate actions on quantities.

Dimensions are physical quantities that can be measured, while units are arbitrary labels that correspond to a 
given dimension to make it relative. 
For example a dimension is length, whereas a \textsf{metre} is a relative unit that describes length.
Units of measure can be defined in the most generic form as either
\textit{base quantities} or \textit{derived quantities}. 
The base quantities are the basic building blocks, and the derived quantities are built from these. The base quantities and derived quantities together form a way of describing any part of the physical world~\cite{sonin2001dimensional}. 
For example length (\textsf{metre}) is a base quantity, and so is time (\textsf{second}). 
If these two base quantities are combined they express velocity ($\mathsf{metre}/\mathsf{second}$ or 
$\mathsf{metre} \times \mathsf{second}^{-1}$) 
which is a derived quantity. 
The International System of Units (SI) defines seven base quantities (length, mass, time,
electric current, thermodynamic temperature, amount of substance,
and luminous intensity) as well as a corresponding unit for each quantity~\cite{nis2015}.
Some popular examples of both base and derived units are shown in Table~\ref{tab:derived}.
It is common for quantities to be declared as a number (the magnitude of the quantity) with an associated 
unit~\cite{bipm-dim}. 

There are many ways in which software processes and development can accommodate units of measurement in this manner~\cite{mcmodelsward21}.
Adding units to conventional programming languages goes back to the 
1970s~\cite{IncoporatingUnits} and early 80s
with proposals to extend Fortran~\cite{gehani} and then Pascal~\cite{dreiheller}. 
Hilfinger~\cite{Hilfinger} showed how to exploit
Ada's abstraction facilities, namely operator overloading and type parameterisation, to 
assign attributes for UoM to variables and values. 
The emergence of object oriented programming languages enabled developers to implement
UoM either through a class hierarchy of units and their derived forms, or through the Quantity 
pattern~\cite{fowler}. There are a large number of libraries for all popular object oriented programming
languages~\cite{unit_oscar18} that support this approach~\cite{unit_mc19}.

Applying UoM annotations requires an advanced checker to ensure variables and method calls
are handled soundly. 
Maxwell introduced the notion of a system of quantities with a corresponding system of units. 
This approach allows scientists working with different measurement systems to communicate more easily~\cite{hall_modelsward}.  
Two units are compatible if they both can be represented as the same derived quantity.
For instance degrees \textsf{Celsius} is compatible with \textsf{Fahrenheit}. 
Values in \textsf{Celsius} can be \textit{converted} to values
in  \textsf{Fahrenheit}, and vice versa. 
This notion of interoperability allows equations to be stated in terms of their dimensions and not their base values.
These can be computed when the equation is evaluated.
Two values can be added or subtracted only if their units are the same. Multiplication
and division either add or subtract the two units product of power representations, assuming both values
are compatible. Although, converting values to ensure compatibility can create round-off errors.
Once a variable has been defined to be of a given unit, then it should remain as such.
Checking that all annotated entities behave according to these rules ensures both \textit{completeness} and
\textit{correctness} of the program, and can be undertaken before the code is run.

However two values that share the same UoM might not represent the same \textit{kinds of quantities} (KOQ)~\cite{marcuspython}.
For example, torque is a rotational force which causes an object to rotate about an axis while work
is the result of a force acting over some distance. 
Surface tension can be described as \textsf{newtons per meter} or
\textsf{kilogram per second squared}, and even though they equate, they represent different quantities. 
Our focus is to present a simple set of rules for arithmetic and function calls that allow quantities to be named and handled \textit{safely}. This is not as straightforward as preserving the names of quantities throughout the programme text.
Multiplication will generate a new quantity so it is very likely that information is lost in intermediate 
stages of a calculation. We propose that functions, whose return KOQ are known, are used to regain information
when calculations use multiplication. In this manner a discipline of programming with quantities is suggested, 
one in which equations involving multiplication are defined as functions to ensure quantities in the main block of code are always known.

This paper extends~\cite{impquant} to include a notion of safe KOQ arithmetic and a demonstration of 
how this discipline deals with information loss, 
along with a streamlined functional presentation of the checking algorithm.
In Section~\ref{background} we describe how UoM are typically implemented and argue for a more
comprehensive representation, while acknowledging the drawbacks of adding complexity to the development process.
In Section~\ref{unitexp} we introduce unit expressions and dimensional analysis.
In Section~\ref{quantities} we describe a simple algebra of named quantities.
and show how they can be maintained while evaluating unit assignments and function calls.
We also discuss some of the obstacles to implementing named quantities and UoM in general.
Finally, in Section~\ref{conc} we summarise quantity validation and describe avenues of current research. 

\section{Background}
\label{background}

One can assert the physical dimension of length with the unit \textsf{metre}
and the magnitude \verb!10! (\verb!10m!). However, the same length can also be expressed using other units such as 
\textsf{centimetres} or \textsf{kilometres}, at the same time changing the magnitude (\verb!1000cm! or \verb!0.01km!). 
Although these examples are all based on the International System of Units (SI)
there exists several other systems, such as the \textit{Imperial system} where \textsf{yards} and \textsf{miles} would be used.
On this basis a very simple object oriented design would entail a superclass for each dimension, such as \verb!Length!, and then specific subclasses for the various units, each of which would contain overloaded operators to 
ensure unit based arithmetic could be performed correctly. 

\medskip

{\small \begin{verbatim}
Length l1 = new LengthMetre (5.0);
Length l2 = new LengthYard (4.0);
Length l3 = l1.addlength (l2);  
\end{verbatim}}

\medskip

\noindent The  \verb!addlength! command would convert \verb!l2! into metres and perform the addition.
We could extend our object oriented design to create a class hierarchy for each base type and use a tree
structure to construct derived types. However this would result in hundreds of units and 
thousands of conversions. 

Fortunately, a normal form exists which makes storage
and comparison a lot easier. Any system of units can be derived from the base units 
as a product of powers of those base units: \verb!base!$^{e_{1}} \times $\verb!base!$^{e_{2}} \times  
\dots $\verb!base!$^{e_{n}}$, 
where the exponents $e_{1},\dots,e_{n}$ are rational numbers. Thus an SI unit can be represented as a 7-tuple 
$\langle e_{1},\dots,e_{7} \rangle$ where $e_{i}$ denotes the $i$-th base unit; or in our case $e_{1}$ denotes
\textsf{length}, $e_{2}$ \textsf{mass}, $e_{3}$ \textsf{time} and so on. 

Performing calculations in relation to quantities, dimensions and units is subtle and can easily lead to mistakes.
A dimensional analysis needs to ensure that (1) two physical quantities can only be equated if they have the same dimensions; (2)  two physical quantities can only be added if they have the same dimensions (known as the \textit{Principle of Dimensional Homogeneity}); 
(3) the dimensions of the multiplication of two quantities is given by the addition of the dimensions of the two quantities. 
If we only consider the three common dimensions of \textsf{length}, \textsf{mass} and \textsf{time}
then we can capture the rules for addition and multiplication. 
\[
\begin{array}{lcll}
(l_{1},m_{1},t_{1}) \dimplus (l_{2},m_{2},t_{2}) & = & (l_{1},m_{1},t_{1}), \mathrm{if} \ l_{1} = l_{2} \wedge m_{1} = m_{2}  \wedge t_{1} = t_{2} \\
(l_{1},m_{1},t_{1}) \dimtimes (l_{2},m_{2},t_{2}) & = & (l_{1}+l_{2},m_{1}+m_{2},t_{1}+t_{2}) \\
\end{array}
\]
The full 7-tuple can be reflected in a typical programming language as an array of integers.
Within an object oriented class structure the array can be coupled with conversion, equality and numeric 
operators to form
a \verb!Unit! abstract data type which ensures only UoM correct arithmetic is undertaken. In Java this would be represented as:

\medskip

{\small
\begin{verbatim} 
class Unit {
  private int [7] dimension;
  private float [7] conversionFactor;
  private int [7] offset;
  ...
  boolean isCompatibleWith (Unit u);
  boolean equals (Unit u);
  Unit multiplyUnits (Unit u);
  Unit divideUnits (Unit u);
}
\end{verbatim}}

\medskip

\noindent This is the basis of the Quantity pattern~\cite{fowler} in which quantity values are represented as a pair: the numerical value, $\{Q\}$, and
the unit of measure, $[Q]$, such that $Q = \{Q\} \cdot [Q]$.

\medskip

{\small
\begin{verbatim}
class Quantity {
  private float value;
  private Unit unit;
  ....
}
\end{verbatim}}

\medskip

\noindent The Quantity pattern provides a means of annotating variable declarations and method
signatures with behavioural UoM specifications. Most libraries for modern programming languages implement this approach 
but, as was found in the survey of~\cite{omar20}, do not satisfy the core requirements of the scientific programming community.
Interview subjects felt that UoM libraries were
inconvenient: they did not interact well with the eco-system, 
incurred additional and often cumbersome syntax, 
had unwanted performance costs as the checking was undertaken at run-time,
required effort to learn and costly rewrites to support.

A value is almost always represented as a float in the Quantity pattern so UoM checking will only apply to that 
number representation. Other number representations will need to be converted and rounding errors absorbed.
Similarly the exponents are typically represented as integers,
but there are rare instances when fractional exponents might be required for intermediate results even
when there are no SI units that requires them.
This inflexibility coupled with an unwieldy syntax and the need for manual conversions are reasons for
the poor adoption of such libraries.

Increasing uptake for quantity aware code requires a language neutral interface that allows programmers to 
manage UoM in an
indistinguishable language agnostic fashion, either as part of the core language (e.g. Swift~\cite{swift} and F\#~\cite{fsharp}) 
or through the use of a separate validator~\cite{Osprey,uomvalidator,cunits,punits}).
Moreover, language based solutions enable UoM checking to be undertaken at compile-time, detecting errors early while ensuring no run-time overheads are required. What these systems fail to provide is a system that checks for
kinds of quantities and ensures entities that have the same UoM but different KOQ are managed safely.

\section{Unit Expressions}
\label{unitexp}

Performing calculations in relation to quantities, dimensions and units is often complex and can easily lead to mistakes.
We shall begin by defining dimensional analysis for a hypothetical programming language extended with unit variable declarations, \textit{udecs},
and only consider the three common dimensions, \textit{dims}. 
Hence \verb!velocity!, namely $\mathsf{length \times time^{-1}}$, is represented as $(1,0,-1)$.
We use the standard \verb!float! implementation to approximate for real numbers but as we are not performing arithmetic any representation would suffice. 
A program consists of a sequence of declarations followed by a sequence of statements.

\bigskip

\noindent \begin{math}
\begin{array}{lcl}
\mathit{prog} & ::= & \tt{begin} \ \mathit{udecs} \; \tt{in} \; \mathit{ustmts} \ \tt{end} \\
\mathit{udecs} & ::= & \mathit{udec}_{1} {\tt ;} \; \ldots {\tt ;} \, \mathit{udec}_{m} \\
\mathit{dims} & ::= &  ( \mathit{int}, \; \mathit{int}, \; \mathit{int}) \\
\mathit{udec} & ::= &  \mathit{uv} \; \verb!: float of !\mathit{dims} \\
\end{array}
\end{math}

\bigskip

It has the standard statement constructs, \textit{ustmt}, and boolean expressions, \textit{bexp},
but we will only focus on quantity variable assignments and conditionals as these affect 
unit variables, \textit{uv}. Further constructs such as while loops would create a more
complete programming language but not add to the presentation.
Unit arithmetic expressions, \textit{uexp}, impose syntactic restrictions so that their soundness can be inferred using the algebra of quantities.
By creating a separate syntax for unit expressions we can distinguish between scalar values and \textit{unitless quantities}, 
namely values that have the dimensions $(0,0,0)$ such as moisture content.

\bigskip

\noindent \begin{math}
\begin{array}{lcl}
\mathit{ustmts} & ::= & \mathit{ustmt}_{1} {\tt ; \;}  \ldots{\tt \; ; } \, \mathit{ustmt}_{m} \\
\mathit{ustmt} & ::= & \mathit{uv} \; \verb!:=!\; \mathit{uexp} \;  \\ 
&  & \lvert \ \tt{if} \  \mathit{bexp} \  \verb!then!  \  \mathit{ustmts}_{1} \  \verb!else! \  \mathit{ustmts}_{2} \\ 
\mathit{uexp} & ::= &  \mathit{uv} \; \lvert \;  \mathit{uexp}_{1}  {\tt +} \;\mathit{uexp}_{2} \; 
\lvert  \;  \mathit{r}  \;{\tt *} \; \mathit{uexp} \;  
\lvert  \;  \mathit{uexp}_{1}  \; {\tt *} \;\mathit{uexp}_{2} \; 
\end{array}
\end{math}

\bigskip

\begin{figure}[ht]
\centering
\begin{math}
\begin{array}{|l|}
\hline \\
\begin{array}{l}
\ \mathcal{P} \strachey{\tt{begin} \ \mathit{udecs} \; \tt{in} \; \mathit{ustmts} \ \tt{end} } \ = \ 
\mathcal{SM}\strachey{\mathit{ustmts}}_{(\mathcal{DS}\strachey{\mathit{udecs}})}
\end{array} \\
\\
\begin{array}{lcl}
\ \mathcal{DS}\strachey{\mathit{udec}_{1}{\tt ;} \; \ldots {\tt ;} \mathit{udec}_{m}} & = &
    \begin{array}[t]{l}
    \mathrm{let} \; \varrho_{1} = \mathcal{D} \strachey{\mathit{udec}_{1}}_{\{ \} } \\
    \verb!     ! \vdots \\
    \mathrm{in} \ \mathcal{D} \strachey{\mathit{udec}_{m}}_{\varrho_{n-1}}
    \end{array} \\

\ \mathcal{D}\strachey{\mathit{uv} \; \tt{: float \ of} \; d}_{\varrho} & = & \varrho \oplus \{ \mathit{uv} \mapsto d \} \\
\end{array} \\
    \\
\begin{array}{l}
\ \mathcal{SM}\strachey{\mathit{ustmt}_{1}\tt{;} \; \ldots \tt{;}  \mathit{ustmt}_{\mathit{m}}}_{\varrho}  \\
\hspace{0.4cm} \  \begin{array}{l} =  
\tt{DimValid}, 
\ \mathrm{if} \ \mathcal{S}\strachey{\mathit{ustmt}_{1}}_{\varrho} = \tt{DimValid} \wedge \, \cdots \, \wedge 
\ \mathcal{S}\strachey{\mathit{ustmt}_{\mathit{m}}}_{\varrho} = \tt{DimValid} \\
 =  \tt{DimFail} , \mathrm{otherwise} \\
  \end{array}
\end{array} \\
\\
\begin{array}{lcl}
\ \mathcal{S}\openS\mathit{uv} \;\verb!:=! \;\mathit{uexp}\closeS_{\varrho} & = &  \tt{DimValid}, \ \mathrm{if} \ \mathcal{UE}\strachey{\mathit{uexp}}_{\varrho} = (\varrho \ \mathit{uv})  \\
 & = &  \tt{DimFail} , \mathrm{otherwise} \\
 \end{array} \\
 \begin{array}{l}
\ \mathcal{S}\strachey{\tt{if} \  \mathit{bexp} \ \tt{then}  \  \mathit{ustmts}_{1} \  \tt{else} \ \mathit{ustmts}_{2} }_{\varrho} \\
 \hspace{0.4cm} \  \begin{array}{l}
      =  \tt{DimValid}, \ \mathrm{if} \ \mathcal{SM}\strachey{\mathit{ustmts}_{1}}_{\varrho} = \tt{DimValid} \wedge 
                \mathcal{SM}\strachey{\mathit{ustmts}_{2}}_{\varrho} = \tt{DimValid} \\
      =  \tt{DimFail}, \ \mathrm{otherwise}
      \end{array} \\
 \end{array} \\ \\
\ \ \mathcal{UE} \ : \ \mathit{uexp} \rightarrow (\mathit{uv} \rightarrow  \mathit{dims}) \rightharpoonup \mathit{dims} \\
 \begin{array}{lcl}
\ \mathcal{UE} \strachey{\mathit{uv}}_{\varrho} & = & \varrho \ \mathit{uv} \\
\ \mathcal{UE} \strachey{\mathit{uexp}_{1} \; \tt{+} \; \mathit{uexp}_{2} }_{\varrho} & = & 
     \mathcal{UE} \strachey{\mathit{uexp}_{1}}_{\varrho} \dimplus  \mathcal{UE} \strachey{\mathit{uexp}_{2} }_{\varrho} \\
\  \mathcal{UE} \strachey{\mathit{r} \; \tt{*} \; \mathit{uexp} }_{\varrho} & = &  \mathcal{UE} \strachey{\mathit{uexp}}_{\varrho} \\
\ \mathcal{UE} \strachey{\mathit{uexp}_{1} \; \tt{*} \; \mathit{uexp}_{2} }_{\varrho} & = & 
     \mathcal{UE} \strachey{\mathit{uexp}_{1}}_{\varrho} \dimtimes  \mathcal{UE} \strachey{\mathit{uexp}_{2} }_{\varrho} \\
 \end{array} \\ \\
 \hline 
 \end{array} 
\end{math}
\caption{Dimensional Analysis rules for declarations, statements and expressions.}
\label{fig:quantityderivation}
\end{figure}

In Figure~\ref{fig:quantityderivation} we present the dimensional analysis rules for programs.
The rules for declarations, $\mathcal{DS}$, build an environment, $\varrho$,
mapping variables to their dimensions. The environment will not change throughout the lifetime of the block.
Thus, once a variable has been defined to be of a given quantity, then it will remain as such. Many library based systems allow
programmers to change the dimensions of unit variables as they are objects of type \verb!Quantity!, namely a mutable array.
Once $\varrho$ has been built, it will be used to perform dimensional analysis on the statements.
The rules for statements, $\mathcal{SM}$, return 
either \verb!DimValid! or \verb!DimFail! depending on whether a given statement uses quantities correctly or not.
An assignment statement is valid only if the quantity of the unit expression is dimensionally homogeneous with
the unit variable that it is being assigned to.
The rule for conditionals checks the dimensional validity of both true and false statements.
The rules for unit expressions, $\mathcal{UE}$, are partial and might not have a solution in the case of trying 
to add quantities that have different dimensions.
The rule for unit variables is just a lookup on the quantity environment $\varrho$.
The rule for addition ensures that both the left hand and right hand side
subexpressions have the same quantities as enforced by the operator $\!\dimplus\!\!$. 
The rules for multiplication allow constants to be applied, and multiplying two unit expressions 
will create a combined quantity, where
each dimension is summed as defined by the operator $\!\!\dimtimes\!\!$.

Dimensional analysis would be sufficient if only one unit system, such as the SI system, was required.
In such cases the base units of \textsf{metre}, \textsf{kilogram} and \textsf{second} could be implicit in implementations. 
Dimensionally correct unit expressions can be evaluated in much the same way as normal arithmetic expressions.
As this is rarely the case in scientific applications where a myriad of unit systems and magnitudes are used, 
we need to perform unit conversions before evaluating the arithmetic expression. 
This can be undertaken at compile-time~\cite{CellmlCooper} or at run-time.
Moreover we need to declare the units alongside their dimensions.
A variable denoting \textsf{torque} would be stored as $\{\verb!t! \mapsto ((\verb!Metre!,2),(\verb!Kilogram!,1),(\verb!Second!,-2))\}$ in the 
environment.

%
%
%

\section{Quantity Rules}
\label{quantities}

This section introduces named quantities, their rules and how they are supported in our typical programming language.
We include functions that provide a clear interface and the potential for more comprehensive checking.

\subsection{Expressions}

We adopt a similar approach to~\cite{Foster2013QuantitiesUA,HallQuant2020} in that quantities should be represented as a 3-tuple, 
and not as a 2-tuple mentioned previously. 
Consequently we add a quantity name, $\langle Q \rangle$,to the numerical value, $\{Q\}$, and
the \textsf{unit} of measure, $[Q]$, such that $Q = \langle Q \rangle \cdot \{Q\} \cdot [Q]$.

However, not all quantity variables in a programme will have a name such as \verb!Torque! or \verb!Work!. Some might denote an entity such
as \textsf{length} that could be in \textsf{metres} or \textsf{yards}, while another might be a variable used to store some temporary value.
Neither of these need to be named. Using an algebraic data type, we define named quantities as:
\[
\verb!type quantname = Named of string | Noname!
\]
We are now in a position to define the rules for adding and multiplying named quantities. In both cases we assume that the unit expression
is dimensionally correct, our concern is to define how named quantities conduct themselves.
The operator \tplusNS $\!$ takes two named quantities and states the conditions under which they can be summed: \textit{two named quantities can be added together only if they 
represent the
same entity}, if one quantity is named but the other is not then it is necessary for the result to be named, and if both are unnamed then
the result will be too:
\[
\begin{array}{lclcll}
\verb!Named!\;n_{1} & \tplus & \verb!Named!\;n_{2}  & = & \verb!Named!\;n_{1}, &\mathrm{if} \ n_{1} = n_{2} \\
\verb!Named!\;n & \tplus & \verb!Noname!  & = & \verb!Named!\;n & \\
\verb!Noname! & \tplus & \verb!Named!\;n  & = & \verb!Named!\;n & \\
\verb!Noname! & \tplus & \verb!Noname!  & = & \verb!Noname! & \\
\end{array}
\]
Our comparison rules cast upwards from \verb!Noname! to \verb!Named!, so as to assume a named quantity whenever possible.
This is required to ensure named quantities behave correctly. If we cast downwards then we would have the alternative
rule $\verb!Named!\;n \tplus \verb!Noname!   =  \verb!Nonamed!$, that would allow  \verb!Work! to be added to \verb!Torque! 
through associativity: (\verb!Named "Work"! \tplus (\verb!Named "Torque"! \tplus \verb!Noname!)) 
$\Rightarrow$ (\verb!Named "Work"! \tplus \verb!Noname!) $\Rightarrow$ \verb!Noname!.

For multiplication the rules are simpler. The operator $\!\!$\ttimes takes in two named quantities and defines how they behave over the multiplication
operator. As \textit{multiplication sums the dimensions of the two operands, the value will be different to either and so the result will always be} \verb!Noname!.
\[
\begin{array}{lclcl}
\verb!Named!\;n_{1} & \ttimes & \verb!Named!\;n_{2} & = & \verb!Noname! \\
\verb!Named!\;n & \ttimes & \verb!Noname! & = & \verb!Noname! \\
 \verb!Noname! & \ttimes & \verb!Named!\;n  & = & \verb!Noname! \\
 \verb!Noname! & \ttimes & \verb!Noname!  & = & \verb!Noname! \\
 \end{array}
\]

\begin{figure}[ht]
\centering
\begin{math}
\begin{array}{|l|}
\hline \\
\ \ \mathcal{NE} \ : \ \mathit{uexp} \rightarrow (\mathit{uv} \rightarrow  \tt{quantname}) \rightharpoonup \tt{quantname} \\
\begin{array}{lcl}
\ \mathcal{NE} \strachey{\mathit{uv}}_{\tau} & = & \tau \ \mathit{uv} \\
 \ \mathcal{NE} \strachey{\mathit{uexp}_{1} \; \tt{+} \; \mathit{uexp}_{2} }_{\tau} & = & 
     \mathcal{NE} \strachey{\mathit{uexp}_{1}}_{\tau} \tplus  \mathcal{NE} \strachey{\mathit{uexp}_{2} }_{\tau} \\
\  \mathcal{NE} \strachey{\mathit{r} \; \tt{*} \; \mathit{uexp} }_{\tau} & = &  \mathcal{NE} \strachey{\mathit{uexp}}_{\tau} \\
\  \mathcal{NE} \strachey{\mathit{uexp}_{1} \; \tt{*} \; \mathit{uexp}_{2} }_{\tau} & = & 
     \mathcal{NE} \strachey{\mathit{uexp}_{1}}_{\tau} \ttimes  \mathcal{NE} \strachey{\mathit{uexp}_{2} }_{\tau} \\
    & & \\
\end{array} \\
 \hline 
 \end{array} 
\end{math}
\caption{Named quantity rules for unit expression.}
\label{fig:namedquantities}
\end{figure}

The named quantity algebra can be incorporated into our language as shown in Figure~\ref{fig:namedquantities}.
The language rules for scaler multiplication do not change the named quantity, the scaler value only affects the quantity value when evaluating expressions.
Consider the example where $\tau = \{\verb!t! \mapsto \verb!Named "Torque"!,  \verb!w! \mapsto \verb!Named "Work"!\}$. If we were
to try to validate $\verb!t! + \verb!w!$ with $\mathcal{NE}$ then the rule for addition will not succeed as there is no case for
(\verb!Named "Torque"!$\tplus$\verb!Named "Work"!).

\begin{theorem}\label{thm1}
For a given unit expression, \textit{uexp}, which does not include general multiplication, and an environment, 
$\tau$, binding unit variables to 
named quantities; $\mathcal{NE}\strachey{\mathit{uexp}}_{\tau}$ will only succeed if all named subexpressions
represent the same entity.
\end{theorem}
\begin{proof}
by induction on unit expressions. The cases for unit variables, $\mathit{uv}$, and scaler multiplication, 
$\mathit{r} \; \tt{*} \; \mathit{uexp}$, are straightforward. The important cases revolve around addition as
we need to consider the potential associative effect of nested subexpressions. Consider 
$\mathit{uexp}_{1}$, $\mathit{uexp}_{2}$ and $\mathit{uexp}_{3}$ where:
\[
\mathcal{NE} \strachey{(\mathit{uexp}_{1} \; \tt{+} \; \mathit{uexp}_{2})  \; \tt{+} \; \mathit{uexp}_{3}}_{\tau}  =  
     (\mathcal{NE} \strachey{\mathit{uexp}_{1}}_{\tau} \tplus  \mathcal{NE} \strachey{\mathit{uexp}_{2} }_{\tau})
     \tplus  \mathcal{NE} \strachey{\mathit{uexp}_{3} }_{\tau}
\]
This will result in the following indicative 5 cases:
\begin{description}
\item[Case 1:] 
\begin{tabular}[t]{l}
(\verb!Named !$n_{1}$$\tplus\,$\verb!Named !$n_{2}$)$\tplus$\verb!Named !$n_{3}$ \\
$\Rightarrow$ \verb!Named !$n_{1}$$\tplus\,$\verb!Named !$n_{3}$, where $n_{1} = n_{2}$ by definition of\tplus\\
$\Rightarrow$ \verb!Named !$n_{1}$, where $n_{1} = n_{2} = n_{3}$ by definition of\tplus\\
\end{tabular}
\item[Case 2:] 
\begin{tabular}[t]{l}
(\verb!Named !$n_{1}$$\tplus\,$\verb!Named !$n_{2}$)$\tplus$\verb!Noname! \\
$\Rightarrow$ \verb!Named !$n_{1}$$\tplus\,$\verb!Noname!, where $n_{1} = n_{2}$ by definition of\tplus\\
$\Rightarrow$ \verb!Named !$n_{1}$, where $n_{1} = n_{2}$ by definition of\tplus\\
\end{tabular}
\item[Case 3:] 
\begin{tabular}[t]{l}
(\verb!Named !$n_{1}$$\tplus\,$\verb!Noname!)$\tplus$\verb!Named !$n_{3}$ \\
$\Rightarrow$ \verb!Named !$n_{1}$$\tplus\,$\verb!Named !$n_{3}$,  by definition of\tplus\\
$\Rightarrow$ \verb!Named !$n_{1}$, where $n_{1} = n_{3}$ by definition of\tplus\\
\end{tabular}
\item[Case 4:] 
\begin{tabular}[t]{l}
(\verb!Noname!$\tplus\,$\verb!Noname!)$\tplus$\verb!Named !$n_{3}$ \\
$\Rightarrow$ \verb!Noname!$\tplus\,$\verb!Named !$n_{3}$, by definition of\tplus\\
$\Rightarrow$ \verb!Named !$n_{3}$,  by definition of\tplus\\
\end{tabular}
\item[Case 5:] 
\begin{tabular}[t]{l}
(\verb!Noname!$\tplus\,$\verb!Noname!)$\tplus$\verb!Noname! \\
$\Rightarrow$ \verb!Noname!$\tplus\,$\verb!Noname!, by definition of\tplus\\
$\Rightarrow$ \verb!Noname!,  by definition of\tplus\\
\end{tabular}
\end{description}
which ensure that \textit{addition} maintains the property that named subexpressions must represent the same entity
for evaluation to succeed.
We will see how to ensure \textit{multiplication} can be made safe in Section~\ref{sub:functioncall}.
\end{proof}

Assignment statements have to satisfy the named quantity of the variable being assigned to, 
specifically the left hand side, and can therefore either succeed or fail.
The $\tassignNS$ rules specify that \textit{one can assign
a named quantity to a variable that has the same named quantity but not otherwise}. 
One can assign a \verb!Noname! value to a named quantity variable as it will have the same dimensions.
However, one cannot allow $\verb!Noname! \tassignNS \;  \verb!Named!\;n$  to $\verb!Succeed!$ as this would 
allow one to assign a \verb!Torque! value to a 
\verb!Work! variable through the intermediary of a local unnamed variable.
The solution is to create new variable bindings and update $\tau$ accordingly.
\[
\verb!type assignstate = Succeed of !(\mathit{uv} \rightarrow \; \verb!quantname!)\verb! | Fail!
\]
\noindent Where $uv$ denotes the variable being assigned to, the rules for $\!\tassign{\tau}{uv}\!$ will override the existing 
binding for $uv$ in the case where we try to assign a named entity to an unnamed variable:
\[
\begin{array}{lclcll}
\verb!Named!\;n_{1} & \tassign{\tau}{uv} & \verb!Named!\;n_{2}  & = &  \verb!Succeed! \;\tau, &\mathrm{if} \ n_{1} = n_{2} \\
\verb!Named!\;n_{1} & \tassign{\tau}{uv}  & \verb!Named!\;n_{2}  & = &  \verb!Fail!, &\mathrm{if} \ n_{1} \neq n_{2} \\
\verb!Named!\;n  & \tassign{\tau}{uv}  & \verb!Noname!  & = &   \verb!Succeed! \;\tau  \\
\verb!Noname!  & \tassign{\tau}{uv} & \verb!Noname!  & = &   \verb!Succeed! \;\tau \\
\verb!Noname! & \tassign{\tau}{uv}  & \verb!Named!\;n & = &   \verb!Succeed! \;\tau\oplus \{uv \; \mapsto \; 
\verb!Named!\;n\} \\
\end{array}
\]

This has a distinct effect on how we define our programming language rules to support the named quantity algebra.
One \textit{must} update the environment $\tau$
to reflect that the named quantity assignment has taken place, as shown in both $\mathcal{NSM}$ and $\mathcal{NS}$ 
rules of Figure~\ref{fig:namedassignments}. For instance, the program:
\medskip
\begin{verbatim}
begin
   t1 : float of Noname;
   t2 : float of Named T
   ...
     t1 := t2
end
\end{verbatim}
\medskip
will update $\tau$ so that \verb!t1! also has the kind \verb!Named "T"!.
This ensures that the bindings of unnamed values will reflect their usage and protect the code from erroneous assignments.
This is \textit{unlike} the rules for dimensional analysis where the environment mapping variables to their dimensions
does not change over their lifetime.
In order to guarantee coherence of potential changes to bindings in $\tau$, the environment is threaded through the rule for
statements and conditionals. 
If an initial \verb!Noname! unit variable is assigned a \verb!Name! then any subsequent attempt to 
redefine it will \verb!Fail!.

\begin{figure}[ht]
\centering
\begin{math}
\begin{array}{|l|}
\hline \\
\begin{array}{l}
\ \mathcal{NSM}\strachey{\mathit{ustmt}_{1}\tt{;} \; \ldots \tt{;}  \mathit{ustmt}_{\mathit{m}}}_{\tau}  \\
\hspace{0.6cm} \  \begin{array}{l} =  
\tt{Succeed}\;\tau_{\mathit{m}}, \ \mathrm{if}
  \begin{array}[t]{l} 
 \ \mathcal{NS}\strachey{\mathit{ustmt}_{1}}_{\tau} = {\tt Succeed}\;\tau_{1} \wedge \ \cdots \ \wedge \ \\
\ \mathcal{NS}\strachey{\mathit{ustmt}_{\mathit{m}}}_{\tau_{m-1}} = {\tt Succeed}\;\tau_{\mathit{m}} \\
 \end{array} \\
 =  \tt{Fail} , \mathrm{otherwise} \\
  \end{array} \\ 
  \\
\end{array} \\
\ \mathcal{NS}\openS\mathit{uv} \; \verb!:=! \; \mathit{uexp}\closeS_{\tau} 
      = (\tau \; \mathit{uv} ) \tassign{\tau}{\mathit{uv}}  (\mathcal{NE} \strachey{uexp}_{\tau}) \\ \\
\ \mathcal{NS}\strachey{\tt{if} \  \mathit{bexp} \ \tt{then}  \  \mathit{ustmts}_{1} \  \tt{else} \ \mathit{ustmts}_{2}}_{\tau} \\ \hspace{0.6cm} \begin{array}[t]{ll}
  = & \tt{Succeed} \;\tau_{2} , \ \mathrm{if} \ 
       \begin{array}[t]{l}\mathcal{NSM}\strachey{\mathit{ustmts}_{1}}_{\tau} = \tt{Succeed}\;\tau_{1}
      \  \wedge \\
       \mathcal{NSM}\strachey{\mathit{ustmts}_{2}}_{\tau_{1}} = \tt{Succeed}\;\tau_{2} \ \ \ \\
       \end{array}\\
  = & \tt{Fail} , \ \mathrm{otherwise} \\
  \end{array} \\
  \\
  \hline
\end{array} 
\end{math}
\caption{Named quantity rules for statements, assignments and conditionals.}
\label{fig:namedassignments}
\end{figure}

\subsection{Function Calls}

Low coupling is generally desirable, especially in large complicated programs which are common nowadays.
Functions enable one to specify a simple interface, to be self-contained, and to be reused. Functions 
are a convenient construct for making pieces of code written by different people or different groups interoperable.

Quantity functions, \textit{ufun}, differ from normal functions in that they can take a number of 
quantity arguments, and return a quantity, if the function body satisfies the quantity algebra. 
Two new syntactic constructs are required:

\bigskip

\noindent \noindent 
\begin{math}
\begin{array}{lcl}
\mathit{ufun} & ::= & \verb!fun! \; \mathit{ufn} {\tt (} uv_{1} {\tt :} \mathit{qn}_{1},\ldots,  uv_{m} {\tt :} \mathit{qn}_{m}{\tt ):} \mathit{qn}_{out}  \\
& & \verb!is! \; \mathit{uexp} \\
\mathit{uexp} & ::= & \ldots \; \lvert \; \mathit{ufn} {\tt (} uexp_{1},\ldots, uexp_{m} {\tt )} \\
\end{array}
\end{math}

\bigskip

\noindent Both a definition mechanism and an invocation mechanism are provided for named quantity functions,
as shown in Figure~\ref{fig:quantitycalls}. Function definitions are stored in a new environment, $\sigma$, binding
function names to their input parameters, expression body and returning quantity. 
Hence our rules for expressions, $\mathcal{NE}$, and assignments will need to be extended to pass 
this second environment around but otherwise stay unchanged.
On invocation we retrieve the definition from $\sigma$, and then build a local unit variable environment, $\tau'$,
which binds the parameters to their quantities.
We must then make sure that named quantities, $\mathit{qn}_{i}$, are safely assigned. 
To be specific, that the named quantity of each argument matches that of its parameter using the
$\!\tassign{\{\}}{uv_{i}}\!$ operator. This will create a single binding if need be, say a parameter is defined in the interface as 
\verb!Noname! but called with a \verb!Name "T"! value.
The initial binding will be overridden to become a value of kind \verb!Name "T"!. On completing the expression body, the
derived named quantity will be compared with that in the definition, using $\!\tplus\!$. Consequently, if the
derived quantity has a different name to that of the function definition, then the analysis cannot proceed. More importantly
though, is if the derived quantity has a \verb!Noname! kind then it can be given a KOQ on return.

\begin{figure} 
\centering
\begin{math}
\begin{array}{|l|}
\hline \\
\ \mathcal{DS}\strachey{\tt{fun} \; \mathit{ufn} \tt{(} \mathit{uv}_{1} \tt{:} \mathit{qn}_{1},\ldots,  \mathit{uv}_{\mathit{m}} 
\tt{:} \mathit{qn}_{\mathit{m}}\tt{):}\mathit{qn}_{out} \; \tt{is} \; \mathit{uexp}}_{\sigma} = \\
\hspace{1cm} \sigma \oplus \{((\mathit{uv}_{1}, \mathit{qn}_{1}),\ldots,  (\mathit{uv}_{m},\mathit{qn}_{m}),(\mathit{uexp},\mathit{qn}_{out}))\} \\
\\
\ \mathcal{NE}\strachey{\mathit{ufn} \tt{(} \mathit{uexp}_{1},\ldots, \mathit{uexp}_{\mathit{m}} \tt{)} }_{\tau\; \sigma} \\
\hspace{1cm} \begin{array}{ll}
= & \mathrm{let} \  ((uv_{1},\mathit{qn}_{1}),\ldots, (uv_{m},\mathit{qn}_{m}),(\mathit{uexp},\mathit{qn}_{out}) ) = \sigma \;  \textit{ufn} \\
&  \mathrm{let} \ \tau' = \{ uv_{1} \mapsto qn_{1}, \; \ldots \; uv_{m} \mapsto qn_{m} \} \\
& \mathrm{in} \ \mathit{qn}_{out} \tplus (\mathcal{NE}\strachey{\mathit{uexp}}_{(\tau' \oplus \,\tau_{1}\oplus\,\cdots\,\oplus \tau_{m})  \; \sigma}), \; \mathrm{if} \\
& \hspace{2cm}
      \begin{array}{l}
      \textit{qn}_{1} \tassign{\{\}}{uv_{1}} (\mathcal{NE}\strachey{\mathit{uexp}_{1}}_{\tau\; \sigma}) 
             = \tt{Succeed} \; \tau_{1} \ \wedge \ \ \ \\
      \hspace{0.6cm} \vdots \\
      \textit{qn}_{m} \tassign{\{\}}{uv_{m}} (\mathcal{NE}\strachey{\mathit{uexp}_{\mathit{m}}}_{\tau\; \sigma})
             = \tt{Succeed}\; \tau_{\mathit{m}} \\
      \end{array}  \\
\end{array} \\ \\
\hline 
\end{array}
\end{math}
\caption{Quantity Checking rules for Function Declarations and Invocation.}
\label{fig:quantitycalls}
\end{figure}


To illustrate how these collection of rules enable named quantity checking we consider a dimensionally correct assignment 
of \verb!nt := 2 * addtq(t1,t2)! with differing named quantity definitions.
In the first case we consider \verb!t1! and \verb!t2! to both
represent torque values, so that environment $\tau$ is 
$\{\verb!nt! \mapsto \verb!Named "T"!, \verb!t1! \mapsto \verb!Named "T"!,  \verb!t2! \mapsto \verb!Named "T"!\}$.
We also define \verb!addtq! to expect two torque quantities.
It will be stored in the function environment $\sigma$
as $\{\verb!addtq! \mapsto ((\verb!x!,\verb!Named "T"!),(\verb!y!,\verb!Named "T"!), (\verb!x+y!,\verb!Named "T"!)) \}$. 
The code fragment would look like this:
\medskip
\begin{verbatim}
begin
 nt : float of Named T;
 t1 : float of Named T;
 t2 : float of Named T;
 fun addtq (x:Named T,y:Named T):Named T = x+y
 ...
   nt := 2 * addtq(t1,t2)
end
\end{verbatim}
\medskip
Quantity checking the assignment would succeed as follows:
\[
\begin{array}{l}
\mathcal{NS}\openS \verb!nt := 2*addtq(t1,t2)! \closeS_{\{ \verb!nt! \mapsto \verb!Named "T"!, \verb!t1! \mapsto \verb!Named "T"!,  \verb!t2! \mapsto \verb!Named "T"!\} \;\sigma} \\
\Rightarrow (\tau \; \verb!nt! ) \tassign{\tau}{{\tt nt}}  (\mathcal{NE} \strachey{{\tt 2 * addtq(t1,t2)}}_{\tau}) \\
\Rightarrow (\verb!Named "T"!) \tassign{\tau}{{\tt nt}}  (\mathcal{NE} \strachey{{\tt addtq(t1,t2)}}_{\tau}) \\
\Rightarrow 
(\verb!Named "T"!) \tassign{\tau}{{\tt nt}} \\
\hspace{1cm}(\verb!Named "T"! \tplus (\mathcal{NE}\strachey{{\tt x+y}}_{(\{\verb!x! \mapsto \verb!Named "T"!,  \verb!y! \mapsto \verb!Named "T"!\} 
\oplus \,\tau_{1}\oplus \tau_{2})  \; \sigma}), \; \mathrm{if} \\
 \hspace{2cm}
      \begin{array}{l}
      \verb!Named "T"!  \tassign{\{\}}{{\tt x}} (\mathcal{NE}\strachey{{\tt t1}}_{\tau\; \sigma}) 
             = \tt{Succeed} \; \tau_{1} \ \wedge \ \ \ \\
      \verb!Named "T"!  \tassign{\{\}}{{\tt y}} (\mathcal{NE}\strachey{{\tt t2}}_{\tau\; \sigma})
             = \tt{Succeed}\; \tau_{2})\\
      \end{array}  \\
\Rightarrow 
(\verb!Named "T"!) \tassign{\tau}{{\tt nt}} \\
\hspace{1cm}(\verb!Named "T"! \tplus (\mathcal{NE}\strachey{{\tt x+y}}_{\{\verb!x! \mapsto \verb!Named "T"!,  \verb!y! \mapsto \verb!Named "T"!\}\;\sigma}), \; \mathrm{if} \\
 \hspace{2cm}
      \begin{array}{l}
      \verb!Named "T"!  \tassign{\{\}}{{\tt x}} \verb!Named "T"!
             = \tt{Succeed} \; \{\} \ \wedge \ \ \ \\
      \verb!Named "T"!  \tassign{\{\}}{{\tt y}} \verb!Named "T"!
             = \tt{Succeed}\; \{\}) \\
      \end{array}  \\
\Rightarrow \verb!Named "T"! \tassign{\tau}{{\tt nt}} (\verb!Named "T"! \tplus \verb!Named "T"!) \\
\Rightarrow \verb!Named "T"! \tassign{\tau}{{\tt nt}} \verb!Named "T"! \\
\Rightarrow \verb!Succeed !\tau
\\
\end{array}
\]
Alternatively, if we try a similar assignment, \verb!nt := 2 * addtq(t,w)!, but with quantities denoting torque and work:
\medskip
\begin{verbatim}
begin
 nt : float of Named T;
 t  : float of Named T;
 w  : float of Named W;
 fun addtq (x:Named T,y:Named T):Named T = x+y
 ...
   nt := 2 * addtq(t,w)
end
\end{verbatim}
\medskip
Such that
$\tau = \{\verb!nt! \mapsto \verb!Named "T"!, \verb!t! \mapsto \verb!Named "T"!,  \verb!w! \mapsto \verb!Named "W"!\}$ 
then the 
derivation cannot be completed as the parameter \verb!w! has the named quantity \verb!Named "W"! where 
a \verb!Named "T"! was expected:
\[
\begin{array}{l}
\mathcal{NS}\openS \verb!nt := 2*addtq(t,w)!\closeS_{\{ \verb!nt! \mapsto \verb!Named "T"!, \verb!t! \mapsto \verb!Named "T"!,  \verb!w! \mapsto \verb!Named "W"!\} \;\sigma} \\
\Rightarrow (\tau \; \verb!nt! ) \tassign{\tau}{{\tt nt}}  (\mathcal{NE} \strachey{{\tt 2 * addtq(t,w)}}_{\tau}) \\
\Rightarrow (\verb!Named "T"!) \tassign{\tau}{{\tt nt}}  (\mathcal{NE} \strachey{{\tt addtq(t,w)}}_{\tau}) \\
\Rightarrow 
(\verb!Named "T"!) \tassign{\tau}{{\tt nt}} \\
\hspace{1cm}(\verb!Named "T"! \tplus (\mathcal{NE}\strachey{{\tt x+y}}_{(\{\verb!x! \mapsto \verb!Named "T"!,  \verb!y! \mapsto \verb!Named "T"!\} 
\oplus \,\tau_{1}\oplus \tau_{2})  \; \sigma}), \; \mathrm{if} \\
 \hspace{2cm}
      \begin{array}{l}
      \verb!Named "T"!  \tassign{\{\}}{{\tt x}} (\mathcal{NE}\strachey{{\tt t}}_{\tau\; \sigma}) 
             = \tt{Succeed} \; \tau_{1} \ \wedge \ \ \ \\
      \verb!Named "T"!  \tassign{\{\}}{{\tt y}} (\mathcal{NE}\strachey{{\tt w}}_{\tau\; \sigma})
             = \tt{Succeed}\; \tau_{2})\\
      \end{array}  \\
\Rightarrow 
(\verb!Named "T"!) \tassign{\tau}{{\tt nt}} \\
\hspace{1cm}(\verb!Named "T"! \tplus (\mathcal{NE}\strachey{{\tt x+y}}_{\{\verb!x! \mapsto \verb!Named "T"!,  \verb!y! \mapsto \verb!Named "T"!\}\;\sigma}), \; \mathrm{if} \\
 \hspace{2cm}
      \begin{array}{l}
      \verb!Named "T"!  \tassign{\{\}}{\tt x} \verb!Named "T"!
             = \tt{Succeed} \; \{\} \ \wedge \ \mathrm{False} \\
       \end{array}  \\
\end{array}
\]
This form of named quantity error detection is labeled Type 1 KOQ error~\cite{marcuspython}.

We can use \verb!Noname! quantities in function interfaces to avoid having to commit to a given name, 
such as torque or work:
\smallskip
\begin{verbatim}
begin
 nt : float of Named T;
 t  : float of Named T;
 w  : float of Named W;
 fun addtq (x:Noname,y:Noname):Noname = x+y
 ...
   nt := 2 * addtq(t,w)
end
\end{verbatim}
\smallskip
In this case our function \verb!addtq! accepts \verb!Noname! quantities so $\sigma = \{\verb!addtq! \mapsto ((\verb!x!,\verb!Noname!),
(\verb!y!,\verb!Noname!), ((\verb!x+y!),\verb!Noname!)\}$. However, the function body will not proceed as both 
arguments to \verb!x+y! need to follow the named quantity rules for addition:
\[
\begin{array}{l}
\mathcal{NS}\openS \verb!nt := 2*addtq(t,w)!\closeS_{\{ \verb!nt! \mapsto \verb!Named "T"!, \verb!t! \mapsto \verb!Named "T"!,  \verb!w! \mapsto \verb!Named "W"!\} \;\sigma} \\
\Rightarrow (\tau \; \verb!nt! ) \tassign{\tau}{{\tt nt}}  (\mathcal{NE} \strachey{{\tt 2 * addtq(t,w)}}_{\tau}) \\
\Rightarrow (\verb!Named "T"!) \tassign{\tau}{{\tt nt}}  (\mathcal{NE} \strachey{{\tt addtq(t,w)}}_{\tau}) \\
\Rightarrow 
(\verb!Named "T"!) \tassign{\tau}{{\tt nt}} \\
\hspace{1cm}(\verb!Noname! \tplus (\mathcal{NE}\strachey{{\tt x+y}}_{(\{\verb!x! \mapsto \verb!Noname!,  \verb!y! \mapsto \verb!Noname!\} 
\oplus \,\tau_{1}\oplus \tau_{2})  \; \sigma}), \; \mathrm{if} \\
 \hspace{2cm}
      \begin{array}{l}
      \verb!Noname!  \tassign{\{\}}{{\tt x}} (\mathcal{NE}\strachey{{\tt t}}_{\tau\; \sigma}) 
             = \tt{Succeed} \; \tau_{1} \ \wedge \ \ \ \\
      \verb!Noname!  \tassign{\{\}}{{\tt y}} (\mathcal{NE}\strachey{{\tt w}}_{\tau\; \sigma})
             = \tt{Succeed}\; \tau_{2})\\
      \end{array}  \\
\Rightarrow 
(\verb!Named "T"!) \tassign{\tau}{{\tt nt}} \\
\hspace{1cm}(\verb!Noname! \tplus (\mathcal{NE}\strachey{{\tt x+y}}_{\{\verb!x! \mapsto \verb!Named "T"!,  \verb!y! \mapsto \verb!Named "W"!\}   \; \sigma}), \; \mathrm{if} \\
 \hspace{2cm}
      \begin{array}{l}
      \verb!Noname!  \tassign{\{\}}{{\tt x}} \verb!Named "T"!
             = \tt{Succeed} \; \{ {\tt x} \mapsto \verb!Named "T"!\} \ \wedge \ \ \ \\
      \verb!Noname!  \tassign{\{\}}{{\tt y}}  \verb!Named "W"!
             = \tt{Succeed}\; \; \{ {\tt y} \mapsto \verb!Named "W"!\} \\
      \end{array}  \\
\Rightarrow 
(\verb!Named "T"!) \tassign{\tau}{{\tt nt}} \\
\hspace{1cm}(\verb!Noname! \tplus (\mathcal{NE}\strachey{{\tt x+y}}_{\{\verb!x! \mapsto \verb!Named "T"!,  \verb!y! \mapsto \verb!Named "W"!\}   \; \sigma}) \\
\Rightarrow 
(\verb!Named "T"!) \tassign{\tau}{{\tt nt}} (\verb!Noname! \tplus (\verb!Named "T"! \tplus \verb!Named "W"!))
\end{array}
\]
This example does show how the local function environment is updated to reflect the kinds of entities passed into it.

\subsection{Safe Multiplication Through Function Calls}
\label{sub:functioncall}

The real potential of quantity functions is that they can re-establish a named quantity. This is of particular interest when
evaluating expressions containing multiplication as
the $\!\ttimes\!$ rules lose information on the kind of quantity generated.
An example of this Type 2 KOQ error~\cite{marcuspython} is the incorrect analysis of a turbine, 
of moment-of-inertia $I$ (SI unit of $kg\cdot m^{2}$) 
rotating with an angular velocity of $\omega_{1}$ ($s^{-1}$) 
with a torque $T$ $(kg\cdot m^{2}\cdot s^{-2})$
applied for duration $t$ in seconds. 
The initial kinetic energy $E_{1}$ is defined as $E_{1} = 0.5 * I * \omega_{1}^{2}$. 
It is easy to code this quantity equation incorrectly 
as $E_{1}= 0.5 * I / t^{2}$, where the units of both sides of the assignment $(kg \cdot m^{2} \cdot s^{-2})$ are compatible 
but the kind of quantity of the unit expression is \verb!Noname!:
\medskip
\begin{verbatim}
begin
 e  : float of Named T;
 i  : float of Named MI;
 t  : float of Named S
  e :=  0.5 * i / (t*t)
end
\end{verbatim}
\medskip
Our dimensional analysis rules would notice the UoM compatibility of the assignment, while our quantity checking rules 
would evaluate the assignment expression, \verb!0.5 * i / (t*t)!, as having \verb!Noname!, so the assignment to \verb!e!,
through $\!\tassign{\tau}{{\tt e}}\!$, would \verb!Succeed!.

However if we demand a discipline of programming with quantities where expressions involving multiplication 
are promoted to functions then we can ensure that results have a known named quantity:
\medskip
\begin{verbatim}
begin
 e  : float of Named T;
 i  : float of Named MI;
 v  : float of Named AV;
 fun kin_energy (I:Named MI,w:Named AV):Named T = 0.5*I*(w*w)
 ...
   e := kin_energy (i,v)
end
\end{verbatim}
\medskip
In this second case, the arguments to \verb!kin_energy! have to represent moment of inertia, \verb!Named MI!, and
angular velocity, \verb!Named AV!, kinds of quantities. On completion the function will return a torque quantity,
\verb!Named T!, so subsequent calculations can be undertaken safely. 
The function behaves like a contract even through we cannot
ascertain the KOQ of \verb!0.5*I*(w*w)! directly using our algebra.

\begin{theorem}\label{thm2}
For a given unit expression, \textit{uexp}, in which general multiplication is undertaken within a function that 
returns a named quantity, along with a variable environment, $\tau$, and a function environment $\sigma$; 
$\mathcal{NE}\strachey{\mathit{uexp}}_{\tau\;\sigma}$ will only succeed if all named subexpressions
represent the same entity.
\end{theorem}
\begin{proof} is an extension of Theorem~\ref{thm1}. The additional base case of the function call
$\mathcal{NE}\strachey{\mathit{ufn} {\tt (} uexp_{1},\ldots, uexp_{m} {\tt )} }_{\tau\;\sigma}$ will, by definition, yield
a named quantity. We therefore need to show the two successful cases that arise when such a function is used within a 
nested subexpression involving another $\mathit{uexp}$:
\[
\begin{array}{l}
\mathcal{NE} \strachey{\mathit{uexp} \; {\tt+} \; \mathit{ufn} {\tt (} uexp_{1},\ldots, uexp_{m} {\tt ) }}_{\tau\;\sigma}  \\
\hspace{2cm} =  
     (\mathcal{NE} \strachey{\mathit{uexp}}_{\tau\;\sigma} )
     \tplus  (\mathcal{NE} \strachey{\mathit{ufn} {\tt (} uexp_{1},\ldots, uexp_{m} {\tt )} }_{\tau\;\sigma})
\end{array}
\]
\begin{description}
\item[Case 1:] 
\begin{tabular}[t]{l}
\verb!Named !$n_{1}$$\tplus\,$\verb!Named !$n_{2}$ \\
$\Rightarrow$ \verb!Named !$n_{1}$, where $n_{1} = n_{2} $ by definition of\tplus \\
\end{tabular}
\item[Case 2:]
\begin{tabular}[t]{l}
\verb!Noname!$\tplus\,$\verb!Named !$n$ \\
$\Rightarrow$ \verb!Named !$n$, by definition of\tplus \\
\end{tabular}
\end{description}
Functions that return a named quantity are safe in their usage of KOQ as subsequent arithmetic can only
proceed if both arguments represent the same quantity.
\end{proof}

\begin{quote}
``We must not forget that it is not our business to make programs, 
it is our business to design classes of computations that will display a desired behaviour.'' (Dijkstra 1972)
\end{quote}

Functions lend themselves towards the development of programming libraries, collections of behaviour, useful for
specific applications which can be used by multiple programs that have no connection to each other.
By extending the interface to include units and kind of quantity information we can ensure that dimensional analysis is
performed and quantities are handled safely. Thereby suggesting that scientific programmers should develop libraries
of functions that are correctly annotated such that checkers can detect misuse of quantities in their code.
Moreover by encouraging library development, the subsequent annotation burden is reduced.

Explicitly named quantity parameters protect the function body but do not allow commonalities to be exploited. 
Having to explicitly name each parameter quantity is cumbersome and minimises reusability as many functions would have to be duplicated. We can use \verb!Noname! quantities to avoid having to commit to a given name, such as torque or work,
but this defeats the purpose of our approach.
Extending named quantities to include  \textit{named quantity variables}, \verb!Quantvar!, 
so that we could write generic quantity functions is preferable as equivalences could be easily defined:
\smallskip
\begin{verbatim}
     fun add (x:Quantvar q,y:Quantvar q):Quantvar q = x+y
\end{verbatim}
\smallskip
In this case the variable \verb!q! could be assigned to any named quantity but both \verb!x! and \verb!y! would have the same named quantity.
Our function invocation rule can be extended to ensure named quantity variables are uniquely assigned, 
and the named quantity of each argument that shared the same quantity variable were equal using a union-find
data structure.

\subsection{Discussion}

Our algebra of named quantities is intended to be implemented as part of the static analysis phase of a compiler for an existing programming language or scientific
domain specific language.
Incorporating named quantities into software models would expand their use, and increase the robustness of designs.
Although named quantity and dimensional analysis are effective at discovering errors early, there are three concerns that impede their 
adoption~\cite{UoMsoftprac20}.
\begin{itemize}
\item \textit{Lack of Awareness:} many developers are totally unaware of software solutions that deal with quantities and UoM.
Inertia arises from factors like tradition, fear of change and effort of learning something new.
Our approach is intended for language extensions or a pluggable type system, 
but could equally be included in popular UoM libraries,
both of which elicit change and adaptation.
\item \textit{Technical Internal Factors:} many solutions are awkward and imprecise, introducing a loss of precision
and struggling at times with dimensional consistencies. The strength of a language based solution, versus a library one, is that
these issues are reduced.
\item \textit{External Factors:} modern systems are not built in a vacuum but form part of an eco-system~\cite{lungoEco}.
It is harder to argue for quantity annotations when values pass through numerous generic
components that do not support them, such as legacy systems, databases, spreadsheets, graphics tools and 
many other components that are unlikely to support quantities without costly updates. 
Efforts are underway to address this essential issue~\cite{codata} but
ensuring that units are routinely documented for easy, unambiguous exchange of information requires
a multi-stakeholder approach~\cite{stopquand}.
\end{itemize}
A lightweight and comprehensive language solution is fundamental to adoption.
However, a relevant observation from both a survey of UoM libraries~\cite{unit_oscar18} and interviews with practitioners~\cite{omar20}
was the need for quantities to be available at run-time. There are many mature and active UoM libraries for popular dynamic object
oriented programming languages, such as Python and Ruby, in which no static checking will occur.
If we allow dimensions to be only available at run-time then quantity checking, and also unit conversion, will have to be undertaken
while the program is executing. 
Faults can be avoided but testing is required to ensure annotations are congruent, 
a feature that a static programme analysis will uncover prior to evaluation.

\section{Conclusion}
\label{conc}

We have developed a simple algebra of named quantities and shown how it can be incorporated into an existing programming language.
Thereby ensuring quantities that share the same units of measure are handled separately, if required. 
Our algebra is safe and allows a degree of flexibility through casting upwards from unnamed quantities to named ones, thus enabling greater code reuse
and a more practical implementation of the concept. Unnamed quantities can be assigned to named one's but not the other way around as
this breaks substitutability. 
Our theorems address both KOQ errors discussed in~\cite{marcuspython} that are not detectable by 
conventional UoM checkers.
The algebra is distinct to that of dimensional analysis but they can be combined to form a single pass validator, allowing 
aliases (such as \verb!J! for work or \verb!N! for force) and standard named mathematical functions to be part of a prelude.

We lack an unqualified estimate of how frequently unit inconsistencies occur or their cost.
Anecdotally we can glean that it is not negligible from experiments described in the 
literature~\cite{CellmlCooper,SpreadsheetValidate,lightweight,ore_uniterrors}.
These focus on checking existing scientific repositories and are not representative of any quantity adhering software discipline.
They have been applied post-development whereas we are seeking to support developers by ensuring their code bases
model their scientific domain.

Quantity checking, much like type checking, is extremely useful when developing as errors are found before any code is run. 
Data on existing systems that have been extensively tested are less indicative of coding practices.
However annotating all unit variables in a programme is costly. 
Ore~\cite{ore_typeburden} found subjects choose a correct
UoM annotation only 51\% of the time and take an average of 136 seconds to make a single correct annotation.
This explains the appeal of systems that try to lower the burden while still ensuring coverage through unit type variables
and a solver~\cite{kennedy94dimension,Osprey,cunits,punits}, 
or compromise coverage through a component based approach~\cite{damevski,lightweight}.
A component based discipline means that the consequences of local unit
mistakes are underestimated. On the other hand, it allows diverse teams to collaborate even if their domain specific environments 
or choice of quantity systems are dissimilar.

If all quantity type variables are resolved by the static checker then dimensional correctness can be shown. 
A corpus of quantity errors, their programming language and associated software systems 
is required to assess the significance of annotations and the cost of developing without.
We have no idea of how much time is spent chasing incompatible quantity assignments, dimension errors or incorrect unit conversions.
We do know that there are many reasons for not adopting a quantity discipline based approach~\cite{omar20}.
Different stakeholders will have different robustness concerns and willingness to compromise on the proportion of 
quantity and unit
annotations required. Addressing usability concerns is an important aspect of our research.

Our tool currently performs static named quantity and dimension analysis for a simple imperative language.
It also searches for the least number of unit conversions to reduce round-off errors in generated code.
The main design principle was based on extending the Quantity pattern with added functionality. 
However, we feel that this strategy does not address the usability concerns of many scientific coders
who require a lightweight bespoke solution to managing UoM as opposed 
to a multi-purpose and inherently cumbersome one.
Consequently we are working on a customisable approach leveraging generics and staged computation.


\section*{Declarations}

Not applicable.

\bibliography{discquant}

\appendix

\end{document}